\documentclass[%
 reprint,
nofootinbib,
 amsmath,amssymb,
 aps,
pra,
]{revtex4-2}
\usepackage{physics}
\usepackage{bbold}
\usepackage{latexsym}
\usepackage{graphicx}
\usepackage{amsthm}
\usepackage{hyperref}
\usepackage{lipsum}
\usepackage[english]{babel}
\usepackage{slashed}
\usepackage{appendix}
\usepackage{subfig}
\usepackage{multirow}
\usepackage{mathtools}
\usepackage{mathrsfs}
\usepackage{amsbsy}
\usepackage{ragged2e}
\usepackage{xcolor}

\newcommand{\hi}{\mathcal{H}}

\newcommand{\Poincup}{\mathcal{P}_+^\uparrow}

\newcommand{\Poincupd}{\overline{\mathcal{P}}_+^\uparrow}
\newcommand{\mink}{\mathbb{M}}

\newcommand{\Bor}{\text{Bor}}

\newcommand{\D}{\mathcal{D}}

\newcommand{\bh}{\mathcal{B}(\hi)}

\newcommand{\ldh}{\mathcal{L}(\D,\hi)}

\newtheorem{theorem}{Theorem} 
\newtheorem*{theorem*}{Theorem} 
\newtheorem{lemma}{Lemma} 
\newtheorem{corollary}{Corollary}

\newtheorem{definition}{Definition}

\theoremstyle{plain}

\newcommand{\sam}[1]{{\color{black}{#1}}}

\usepackage[bottom]{footmisc}

\begin{document}

\title{No-go theorems for pointwise-defined spinorial quantum fields}
\author{Samuel Fedida}
\affiliation{Centre for Quantum Information and Foundations, DAMTP, Centre for Mathematical Sciences, University of Cambridge, Wilberforce Road, Cambridge CB3 0WA, UK}
\date{\today}

\begin{abstract}
    We extend and strengthen no-go results on pointwise-defined quantum fields to cover general spinors. We show that the weak continuity of quantum fields rules out equal-time canonical conjugate (anti)commutation relations in globally hyperbolic spacetimes; for quantum fields on Minkowski spacetime, weakly continuous translation covariance enforces the needed continuity and yields the same no-go. \sam{We then prove a fermionic microcausality no-go result: a weakly continuous pointwise fermionic field satisfying local spacelike anticommutation with its adjoint field on a $C^2$ Lorentzian spacetime must vanish.} We finish by generalising Wizimirski's no-go theorem to show that the existence of a Poincaré-invariant \sam{separating} vacuum precludes pointwise spinorial covariance on a Minkowski background. \sam{The result applies to Weyl and Dirac multiplets and to gauge-invariant field-strength multiplets such as the electromagnetic field strength and the linearised Weyl curvature. Gauge potentials are covered only when exact tensor covariance, rather than covariance modulo gauge transformations, is imposed.}
\end{abstract}

\maketitle

\section{Introduction}

The mathematical formulation of quantum field theory confronts deep challenges when attempting to define quantum fields as operator-valued functions on spacetime. In particular, several no-go theorems -- originating in the works of Wightman \cite{wightman_theorie_1964} and Wizimirski \cite{wizimirski_existence_1966} and further developed in algebraic quantum field theory \cite{Halvorson2006col} -- demonstrate that under natural physical assumptions, pointwise-defined fields often lead to contradictions. While these results are well established for scalar fields, the cases of other spinorial fields requires separate scrutiny due to the different algebraic structures they obey, and to check which assumptions are fundamentally causing issues. 

This note extends no-go results to the context of more general spinorial quantum fields. We first review a result originally presented for bosonic canonical commutation relations on Minkowski spacetime (e.g. see \cite{Halvorson2006col}), and show that the weak continuity of quantum fields on globally hyperbolic Lorentzian spacetimes rules out both equal-time canonical commutation relations and canonical anticommutation relations. Consequently, in Minkowski spacetime, space translation covariance with respect to a weakly continuous representation\footnote{The assumption of the weak continuity of $U$ is mild and is ensured by assuming the existence of self-adjoint momentum generators $P_i$ using Stone's theorem \cite{stone_one-parameter_1932} (or by the SNAG theorem \cite{streater_pct_1989} in relativistic settings). Note that for a unitary representation of the translation group, strong and weak continuity are equivalent.} alone straightforwardly rules out the possibility of realising equal-time conjugate commutation relations and conjugate anticommutation relations for any pointwise-defined quantum field. 

\sam{Motivated by Wightman's no-go theorem \cite{wightman_theorie_1964}, but using a different and simpler argument, we prove} that fermionic microcausality implies that any weakly continuous fermionic quantum field on time-oriented $C^2$ Lorentzian spacetimes must vanish. This implies that quantum fields on Minkowski that are translation-covariant with respect to a weakly continuous unitary representation cannot be fermionically microcausal. This, in turn, provides a deeper justification for the rejection of microcausality: other usual assumptions (the spectrum condition, translation covariance, the uniqueness of the vacuum) are not necessary to rule out such an ontic (rather than epistemic) implementation of the no-superluminal signalling principle. 

We finish by extending Wizimirski's theorem \cite{wizimirski_existence_1966} to other spinors, proving that the existence of a Poincaré-invariant \sam{separating} vacuum prohibits pointwise Poincaré covariance for multiplicity-free nontrivial finite-dimensional representations of the spin group in $1+3$ dimensions. This covers Weyl and Dirac spinors, photons as well as gravitons on a Minkowski background. 

The statements below are intended for pointwise-defined fields. In standard axiomatic quantum field theory, the fields are operator-valued distributions from the outset and are only defined after smearing with test functions \cite{wightman_theorie_1964,streater_pct_1989,Halvorson2006col}. \sam{One does not have an operator $\hat{\phi}(x)$ for each point $x \in \mathcal{M}$; rather, one has smeared operators $\hat{\Phi}(f)$ with $f \in C^\infty_c(\mathcal{M})$. The limiting arguments used below cannot be run in the smeared setting, because point evaluation would require testing the distribution against
$\delta_x$, which is not an element of $C_c^\infty(\mathcal M)$. Correspondingly, canonical relations of Wightman quantum fields take a smeared form, and spacelike (anti)commutativity is imposed for test functions with spacelike-separated supports rather than pointwise.} \\

Such distributional fields thus evade the point-wise no-go theorems below. We emphasise the present results simply as obstructions to the existence of genuine operator-valued functions on spacetime that obey the canonical (anti)commutation relations, the spacelike (anti)commutativity constraints or the pointwise covariance properties; they do not contradict the existence of usual smeared quantum fields such as the free scalar, Dirac or Maxwell fields in Wightman or algebraic quantum field theory. They also do not cover the construction of quantum fields as sesquilinear forms \cite{rehberg_quantum_1986, wollenberg_quantum_1986}.

We briefly mention one recently introduced axiomatic approach to quantum field theory which does manipulate covariant and causal pointwisely-defined quantum fields, namely relational quantum field theory \cite{fedida_foundations_2025}. The no-go theorems discussed below are however avoided in somewhat subtle ways. For instance, the pointwisely-defined relational quantum fields are \emph{not} generally weakly continuous, so microcausality can be imposed on those under certain assumptions.
Likewise, Poincaré covariance can be imposed pointwise with, however, a shift in the quantum reference frame describing the fields; thus, the \emph{operator-valued functions} are \emph{not} pointwisely covariant. The no-go results presented in this note can hence also be seen as useful tools for theory building to understand which model-specific assumptions concerning causality and covariance can be safely considered. 

\sam{These no-go results are also distinct from Haag's theorem \cite{haag_quantum_1955,streater_pct_1989}. Haag's theorem concerns the incompatibility of an interaction-picture unitary equivalence between free and interacting fields under the usual relativistic assumptions. The results below do not compare two theories or two representations and do not rule out interactions. They isolate a different obstruction at a functional analytic level: the attempt to impose canonical, causal, or covariance relations directly on fields treated as pointwise operator-valued functions rather than as operator-valued distributions leads to issues. Thus the standard distributional formulation of QFT evades the present pointwise no-go theorems, while Haag's theorem addresses a separate obstruction related to the interaction picture (and, in particular, does apply also to distributional fields).}

\sam{The paper is structured as follows. Section \ref{sec:ccr} highlights an obstruction to canonical equal-time (anti)commutation relations for weakly continuous pointwise quantum fields on globally hyperbolic spacetimes. We review the scalar Wightman and Wizimirski theorems in Sections \ref{sec:microcausality} and \ref{sec:covariance} to compare them with the new results of this paper. The new contributions are Thm.~\ref{thm:Wight}, which gives an obstruction to fermionic microcausality on $C^2$ Lorentzian spacetimes without assuming a spectrum condition, translation covariance, nor the existence of a vacuum; and Thms.~\ref{thm:chiral} and \ref{thm:gravitons}, which give obstructions to Poincaré (spinor) covariance for finite-dimensional $SL(2,\mathbb C)$ multiplets.}

In this paper, we work with $d$-dimensional spacetimes in mostly minus signature, $d \geq 2$, except in the last section where we restrict ourselves to $1+3$ spacetime dimensions. We write $(\mathcal{M},g)$ for a Lorentzian spacetime and $\mink$ for Minkowski spacetime. We denote by $\bh$ the space of bounded operators on a Hilbert space $\hi$ and $\ldh$ the linear maps from a dense domain $\D \subset \hi$ to $\hi$. Given a unitary representation $U$ of a group $G$ on $\hi$ we write $G \stackrel{U}{\curvearrowright} \hi$. We write $\comm{A}{B} = AB - BA$ and $\acomm{A}{B} = AB + BA$. We denote the proper orthochronous Poincaré group by $\Poincup$ and its double cover by $\Poincupd$. \sam{Throughout Sec.~\ref{sec:covariance}, ``spinor" is used in the representation-theoretic sense of a finite-dimensional $SL(2,\mathbb C)$ multiplet. This includes ordinary spinors as well as tensor representations.}

\section{Canonical (Anti)Commutation Relations}
\label{sec:ccr}

We here follow the proof from \cite{Halvorson2006col} which covered the bosonic case of conjugate commutation relations, with almost no alterations to cover the fermionic case in globally hyperbolic spacetimes beyond translation covariance -- we simply include the following result for completeness.

\begin{theorem}
    \label{thm:can anti rel}
    Let $(\mathcal{M},g)$ be a globally hyperbolic Lorentzian spacetime \sam{of dimension $d \geq 2$} and $\{\Sigma_t\}
    _{t \in \mathbb{R}}$ be a foliation of $\mathcal{M}$ into spacelike Cauchy hypersurfaces, and $\D$ be a dense subset of $\hi$. Let $\hat{\psi} : \mathcal{M} \to \ldh$ and $\hat{\pi} : \mathcal{M} \to \ldh$ be such that\footnote{Each $\hat{\psi}(x)$ can, in principle, have different domains, so we consider here their restriction to their joint domain $\D$ which must be dense in $\hi$, and likewise for $\hat{\pi}$.} 
    \begin{enumerate}
        \item $\hat{\psi}(x) \D \subset \D$ and $\hat{\psi}(x)^\dagger \D \subset \D$ for all $x \in \mathcal{M}$,\footnote{This assumption, though restrictive, is common e.g. in constructive quantum field theory \cite{streater_pct_1989}.}
        \item $\hat{\pi}(x) \D \subset \D$ and $\hat{\pi}(x)^\dagger \D \subset \D$ for all $x \in \mathcal{M}$,
        \item the restriction $\hat{\pi}|_{\Sigma_t}$ is weakly continuous\footnote{That is, the map $x \in \Sigma_t \mapsto \bra{\phi}\hat{\pi}(x)\ket{\psi}$ is continuous for all $\ket{\phi},\ket{\psi} \in \D$ and all $t \in \mathbb{R}$.} $\forall t \in \mathbb{R}$.
    \end{enumerate}
    Then for $t \in \mathbb{R}$ and $x \in \Sigma_t$,
    \begin{multline*}
        \comm{\hat{\psi}(x)}{\hat{\pi}(y)} = 0 \quad \forall y \in \Sigma_t \smallsetminus \{x\} \\ \Leftrightarrow \comm{\hat{\psi}(x)}{\hat{\pi}(y)} = 0 \quad \forall y \in \Sigma_t\, 
    \end{multline*}
    and
    \begin{multline*}
        \acomm{\hat{\psi}(x)}{\hat{\pi}(y)} = 0 \quad \forall y \in \Sigma_t \smallsetminus \{x\} \\ \Leftrightarrow \acomm{\hat{\psi}(x)}{\hat{\pi}(y)} = 0 \quad \forall y \in \Sigma_t\, ,
    \end{multline*}
    where these (anti)commutators are understood on $\D$.
\end{theorem}

\begin{proof}
    Let us consider the anticommutation case -- the commutation case follows analogously. The $\Leftarrow$ direction is immediate, so we consider the other direction. Let $t \in \mathbb{R}$ and $x \in \Sigma_t$. The map
    \begin{align}
        f_{x} : \Sigma_t &\to \ldh \nonumber \\
        f_{x}(y) &:= \acomm{\hat{\psi}(x)}{\hat{\pi}(y)} \nonumber 
    \end{align}
    is weakly continuous on $\Sigma_t$ by domain invariance and since $\hat{\pi}|_{\Sigma_t}$ is weakly continuous. Let $(y_n)_{n \in \mathbb{N}}$ be a sequence in $\Sigma_t \smallsetminus \{x\}$ such that $y_n \to x$. Since $f_{x}$ is weakly continuous and $f_x(y_n) = 0$ for all $n \in \mathbb{N}$, we have
    \begin{equation}
        \acomm{\hat{\psi}(x)}{\hat{\pi}(x)} = f_x(x) = \lim_{n \to \infty} f_x(y_n) = 0 \, , 
    \end{equation}
    which concludes the proof.
\end{proof}

Note that the proof also holds in other topologies, notably if $\hat{\pi}$ is strongly continuous. A direct corollary of the above is the fermionic counterpart of the result in \cite{Halvorson2006col} (restricted to bounded quantum fields for conciseness).

\begin{corollary}
    Suppose $d \geq 2$ and $\mathbb{R}^{d-1} \stackrel{U}{\curvearrowright} \hi$. Let $\hat{\psi} : \mink \to \bh$ and $\hat{\pi} : \mink \to \bh$ be such that
    \begin{enumerate}
        \item $U$ is a weakly continuous unitary representation of $\mathbb{R}^{d-1}$,
        \item $U(\vec{y}) \hat{\pi}(t,\vec{x}) U(\vec{y})^\dagger = \hat{\pi}(t,\vec{x} +\vec{y})$ $\forall \vec{x},\vec{y} \in \mathbb{R}^{d-1}, t \in \mathbb{R}$.
    \end{enumerate}
     Then
    \begin{multline*}
        \comm{\hat{\psi}(t,\vec{x})}{\hat{\pi}(t,\vec{y})} = 0 \quad \forall \vec{x}\neq \vec{y} \in \mathbb{R}^{d-1} \\ \Leftrightarrow \comm{\hat{\psi}(t,\vec{x})}{\hat{\pi}(t,\vec{y})} = 0 \quad \forall \vec{x}, \vec{y} \in \mathbb{R}^{d-1} 
    \end{multline*}
     and
    \begin{multline*}
        \acomm{\hat{\psi}(t,\vec{x})}{\hat{\pi}(t,\vec{y})} = 0 \quad \forall \vec{x}\neq \vec{y} \in \mathbb{R}^{d-1} \\ \Leftrightarrow \acomm{\hat{\psi}(t,\vec{x})}{\hat{\pi}(t,\vec{y})} = 0 \quad \forall \vec{x}, \vec{y} \in \mathbb{R}^{d-1} \, .
    \end{multline*}
\end{corollary}

\begin{proof}
    The map $\vec{y} \mapsto \bra{\phi}U(\vec{y})\hat{\pi}(t,\vec{x}) U(\vec{y})^\dagger\ket{\psi}$ is continuous for all $\ket{\phi},\ket{\psi} \in \hi$ and all $(t,\vec{x}) \in \mink$ by the weak continuity of $U$, so the result follows from Thm. \ref{thm:can anti rel}.
\end{proof}

Thus, this no-go theorem is not based on some tension between quantum mechanics and relativity through some possibly problematic notion of ``equal times" since this holds for non-relativistic quantum fields \cite{Halvorson2006col}. 

Hence, the equal-time canonical conjugate (anti)commutation relations for translation-covariant fermionic (or bosonic) weakly continuous spinors
\begin{equation}
    \label{eqn:car}\acomm{\hat{\psi}_\alpha(t,\vec{x})}{\hat{\pi}_\beta(t,\vec{y})} = i \delta_{\alpha \beta} \delta^{(d-1)}(\vec{x}-\vec{y})
\end{equation}
cannot be satisfied. \sam{After taking matrix elements, the left-hand side is a continuous function of $\vec{y}$ under the hypotheses above. If it vanishes for $\vec{x}\neq\vec{y}$, Thm.~\ref{thm:can anti rel} forces it to vanish also at $\vec{x}=\vec{y}$. The right-hand side, however, contains the distribution $\delta^{(d-1)}(\vec{x}-\vec{y})$ which is not a continuous function, and vanishes for $\vec{x}\neq\vec{y}$ but not (distributionally) for $\vec{x}=\vec{y}$. This is precisely why the canonical relations are naturally formulated only after smearing with respect to a test function.}
This is also why in some axiomatic approaches to quantum field theory such as Wightman quantum field theory \cite{wightman_theorie_1964,streater_pct_1989}, quantum fields are instead seen as operator-valued distributions. 

\section{Microcausality}
\label{sec:microcausality}

Microcausality is one (out of many) quantum-mechanical implementations of the relativistic no-signalling principle. It imposes what is fundamentally an epistemic principle (no-signalling) to be applied at an ontological level (at least, provided the field operators carry ontological weight). The no-go theorem by Wightman \cite{wightman_theorie_1964,Halvorson2006col}, subsequently adapted to more general algebraic settings \cite{horuzhy_axiomatic_1990}, highlights that this is fundamentally inconsistent in the context of translation covariant bosonic quantum fields, as follows.

\begin{theorem*}[Wightman]
    Suppose $\mathbb{R}^d \stackrel{U}{\curvearrowright} \hi$. Let $\hat{\phi} : \mink \to \bh$.\footnote{The theorem is here expressed for bounded quantum fields, though an extension to the unbounded case is straightforward given appropriate domain and continuity assumptions.} If
	\begin{enumerate}
        \item $U$ is a weakly continuous unitary representation of $\mathbb{R}^d$,
		\item $\comm{\hat{\phi}(x)}{\hat{\phi}(y)} = 0$ if $(x-y)^2<0$,
		\item $U(y) \hat{\phi}(x) U(y)^\dagger = \hat{\phi}(x+y)$ for all $x,y \in \mink$,
		\item $U$ satisfies the spectrum condition\footnote{The spectrum condition states that $\sigma(P_\mu) \subset \overline{V_+} = \{p \mid p^0 \geq 0, p^2 \geq 0\}$, i.e. the joint spectrum $\sigma(P_\mu)$ of the energy-momentum lies in the forward causal cone.},
		\item There exists a unique (up to scalar multiples) non-zero vector $\ket{\Omega} \in \hi$ such that $U(x) \ket{\Omega} = \ket{\Omega}$ for all $x \in \mink$.
	\end{enumerate}
	Then there exists a $c \in \mathbb{C}$ such that $\forall x \in \mink$,  \[ \hat{\phi}(x) \ket{\Omega} = c \ket{\Omega} \, .\] 
\end{theorem*}

It could be argued that the spectrum condition, the existence and uniqueness of a vacuum or even pointwise translation covariance could fail to hold exactly, whilst retaining the validity of microcausality. For example, one would not expect both translation covariance and the uniqueness of a vacuum to hold in curved backgrounds. Here, we show that this line of reasoning is not right for the case of fermions: the weak continuity of the quantum fields is sufficient to rule out fermionic microcausality.

\sam{\begin{definition}
    We say that two points $x,y\in\mathcal M$ are \emph{locally spacelike separated} if they lie in a common geodesically convex normal neighbourhood and the unique geodesic segment connecting them in that neighbourhood has spacelike tangent.
\end{definition}}

\begin{theorem}
    \label{thm:Wight}
    Let $(\mathcal{M},g)$ be a $C^2$ Lorentzian spacetime of dimension $d \geq 2$ and $\D$ be a dense subset of $\hi$. Let $\hat{\psi} : \mathcal{M} \to \ldh$ be weakly continuous at every point. Suppose $\hat{\psi}(x) \D \subset \D$ and $\hat{\psi}(x)^\dagger \D \subset \D$ for all $x \in \mathcal{M}$. Then \sam{
    \begin{align*}
        \acomm{\hat{\psi}(x)^\dagger}{\hat{\psi}(y)} &= 0 \, \; \forall x,y \in \mathcal{M} \text{ locally spacelike separated} \\
        \Leftrightarrow \hat{\psi}(x) &= 0 \; \forall x \in \mathcal{M},
    \end{align*}
    where these equalities are understood on $\D$.}
\end{theorem}

\begin{proof}
    The $\Leftarrow$ direction is immediate, so we consider the other direction. Let $x \in \mathcal{M}$ and $\mathbf{v} \in T_x\mathcal{M}$ be any spacelike vector ($g_x(\mathbf{v},\mathbf{v}) < 0$). \sam{Choose a geodesically convex normal neighbourhood $O$ of $x$. For all sufficiently small $\epsilon>0$, $x_\epsilon:=\exp_x(\epsilon v)$ lies in $O$ and is locally spacelike separated from $x$. Let $\epsilon_n\downarrow0$ and $x_n=x_{\epsilon_n}$. Then $x_n \to x$ and each $x_n$ is locally spacelike separated} to $x$ (for $n$ sufficiently large -- without loss of generality for all $n \in \mathbb{N}$). Thus,
    \begin{equation}
        \acomm{\hat{\psi}(x)^\dagger}{\hat{\psi}(x_n)} = 0 \qquad \forall n \in \mathbb{N}
    \end{equation}
    so by domain invariance and since addition is weakly continuous on $\D$, we have
    \begin{equation}
        \lim_{n \to \infty} \acomm{\hat{\psi}(x)^\dagger}{\hat{\psi}(x_n)} = \acomm{\hat{\psi}(x)^\dagger}{\hat{\psi}(x)} = 0 \, 
    \end{equation}
    weakly.
    Thus, for all $\ket{\chi} \in \D$, taking the expectation value of both terms with respect to $\ket{\chi}$ yields
    \begin{equation}
        \norm{\hat{\psi}(x)^\dagger \ket{\chi}}^2 + \norm{\hat{\psi}(x)\ket{\chi}}^2 = 0 \, .
    \end{equation}
    But norms are non-negative, so 
    \begin{equation}
        \norm{\hat{\psi}(x)^\dagger \ket{\chi}}^2 = \norm{\hat{\psi}(x) \ket{\chi}}^2 = 0
    \end{equation}
    and hence $\hat{\psi}(x) \ket{\chi} = \hat{\psi}(x)^\dagger \ket{\chi} = 0$ for all $\ket{\chi} \in \D$ and all $x \in \mathcal{M}$ which concludes the proof.
\end{proof}

Again, this proof also follows if one works in the strong operator topology. A corollary is that nonzero translation-covariant (with respect to a weakly continuous unitary representation) pointwise quantum fields cannot be fermionically microcausal on Minkowski spacetime. We restrict our attention to bounded quantum fields for conciseness.

\begin{corollary}
	Suppose $\mathbb{R}^{d-1} \stackrel{U}{\curvearrowright} \hi$. Let $\hat{\psi} : \mink \to \bh$. If
	\begin{enumerate}
        \item $U$ is a weakly continuous unitary representation of $\mathbb{R}^{d-1}$,
		\item $U(\vec{x}) \hat{\psi}(t,\vec{0}) U(\vec{x})^\dagger = \hat{\psi}(t,\vec{x})$ $\forall t \in \mathbb{R}$ and $\forall \vec{x} \in \mathbb{R}^{d-1}$,
	\end{enumerate}
    Then
    \begin{align*}
        \acomm{\hat{\psi}(x)^\dagger}{\hat{\psi}(y)} &= 0 \qquad \forall x,y \in \mink \text{ s.t. } (x-y)^2<0 \\ \Leftrightarrow \hat{\psi}(x) &= 0 \qquad \forall x \in \mink \, .
    \end{align*}
\end{corollary}

\begin{proof}
    Pick $\mathbf{v} \in \mathbb{R}^{d-1}$ spacelike and let $x_n = x + \left(0,\frac{\mathbf{v}}{n}\right)$ in some coordinate chart. Then $x_n \to x$ and $\forall n \in \mathbb{N}$, 
    \begin{equation}
        \acomm{\hat{\psi}(x)^\dagger}{\hat{\psi}(x_n)} = 0 
    \end{equation}
    so, by continuity,
    \begin{equation}
        \lim_{n \to \infty} \acomm{\hat{\psi}(x)^\dagger}{\hat{\psi}(x_n)} = \acomm{\hat{\psi}(x)^\dagger}{\hat{\psi}(x)} = 0 \, 
    \end{equation}
    As in Thm. \ref{thm:Wight}, this implies that $\hat{\psi}(x) \ket{\chi} = \hat{\psi}(x)^\dagger \ket{\chi} = 0$ for all $\ket{\chi} \in \hi$ and all $x \in \mink$.
\end{proof}

This confirms the idea that microcausality, whether bosonic or fermionic, is too strong to be applied to translation-covariant \sam{pointwise} quantum fields. Indeed, algebraic \cite{horuzhy_axiomatic_1990,Halvorson2006col} and Wightmannian \cite{streater_pct_1989} methods typically solve such issues by replacing microcausality by the weaker Einstein causality postulate (quantum measurements over spacelike-separated regions commute) -- and indeed pointwise translation-covariance by smeared covariance. 

\section{Poincaré covariance}
\label{sec:covariance}

We finish with a generalisation of Wizimirski's theorem \cite{wizimirski_existence_1966,Halvorson2006col} to cover pointwise Poincaré covariance. This notion of covariance should fundamentally be understood as an epistemic rather than ontological statement. Indeed, the requirement that ``the physics \emph{looks} the same in every inertial reference frame" is a principle based on  epistemology (e.g. at the level of measurements), not on the underlying reality. Wizimirski's theorem again highlights this fact by showing that it can be too strong a requirement to impose at the pointwise level for scalar quantum fields.

\begin{theorem*}[Wizimirski]
    Suppose $\Poincup \stackrel{U}{\curvearrowright} \hi$. Let $\hat{\phi} : \mink \to \bh$. If
    \begin{enumerate}
        \item $U$ is a weakly continuous unitary representation of $\Poincup$,
        \item For all $(a,\Lambda) \in \Poincup$, $x \in \mink$,
        \[ U(a,\Lambda) \hat{\phi}(x) U(a,\Lambda)^\dagger = \hat{\phi}(\Lambda x + a) \, ,\]
        \item There exists a unique (up to scalar multiples) nonzero vector $\ket{\Omega} \in \hi$ such that $U(x,e) \ket{\Omega} = \ket{\Omega}$ for all $x \in \mink$.
    \end{enumerate}
    Then $\exists c \in \mathbb{C}$ such that $\hat{\phi}(x) \ket{\Omega} = c \ket{\Omega}$ for all $x \in \mink$.
\end{theorem*}

One might worry that assuming a unique translation-invariant vacuum limits physical relevance, since symmetry breaking can lead to vacuum degeneracy. Below, we will be able to weaken this assumption by instead assuming the \emph{existence} (rather than uniqueness) of a Poincaré-invariant vector (though now the restriction is on the stronger group invariance requirement).

We here show that the no-go theorem holds for more general spinors on $1+3$-dimensional Minkowski spacetime. We first rule out pointwise Poincaré covariance for spinors transforming nontrivially under $SL(2,\mathbb{C})$ (assuming no sub-representation multiplicity for simplicity).

\begin{theorem}
    \label{thm:chiral}
    Suppose $\Poincupd \stackrel{U}{\curvearrowright} \hi$ and let $\D$ be a dense subset of $\hi$.  Let $S : SL(2,\mathbb{C}) \to GL(\mathbb{C}^n)$ be a multiplicity-free finite-dimensional representation of $SL(2,\mathbb{C})$ with no trivial sub-representations and $\hat{\psi} = \{\hat{\psi}_\alpha : \mink \to \ldh\}_{\alpha=1}^{n \in \mathbb{N}}$. If
    \begin{enumerate}
        \item $U$ is a unitary representation of $\Poincupd$ such that for all $(a,A) \in \Poincupd$, $U(a,A) \D \subset \D$ and $U(a,A)^\dagger \D \subset \D$,
        \item $\hat{\psi}_\alpha(x) \D \subset \D$ and $\hat{\psi}_\alpha(x)^\dagger \D \subset \D$ for all $x \in \mink$ and for all $\alpha \in \{1,\cdots,n\}$,
        \item For all $(a,A) \in \Poincupd, \, x \in \mink, \, \alpha \in \{1,\cdots,n\}$, \[ U(a,A) \hat{\psi}_\alpha(x) U(a,A)^\dagger = \sum_{\beta=1}^n S[A^{-1}]_{\alpha \beta} \hat{\psi}_\beta(\Lambda(A) x + a) \, , \]
        which is understood on $\D$,
        \item There exists a non-zero vector $\ket{\Omega} \in \D$ such that $U(a,A) \ket{\Omega} = \ket{\Omega}$ for all $(a,A) \in \Poincupd$.
    \end{enumerate}
    Then $\forall x \in \mink$ and all $\alpha \in \{1,\cdots,n\}$, \[ \hat{\psi}_\alpha(x) \ket{\Omega} = \hat{\psi}_\alpha(x)^\dagger \ket{\Omega} =  0 \, .\]
    \sam{In particular, if $\ket{\Omega}$ is separating for the $*$-algebra generated by the $\hat{\psi}_\alpha$ and the $\hat{\psi}_\alpha^\dagger$, then $\forall x \in \mink$ and all $\alpha \in \{1,\cdots,n\}$,
    \[ \hat{\psi}_\alpha(x)  = \hat{\psi}_\alpha(x)^\dagger  =  0 \]
    on $\D$.
    }
\end{theorem}

\begin{proof}    
    See Appendix \ref{app:nontrivial rep}.
\end{proof}

Note that we did not require any form of continuity of $U$ in the above, nor did we assume microcausality. \sam{This applies, for example, to Weyl and Dirac spinor multiplets and to the gauge-invariant electromagnetic field-strength multiplet. Gauge potentials are covered only when the exact covariance law above is imposed rather than covariance modulo gauge transformations. Indeed, these gauge fields are representatives of gauge-equivalence classes, so any covariance relation can, in principle, relate two different members of the same equivalence class, and thus not satisfy the no-go theorem.} Still, we can further expand on this to cover gravitons, which transform as $(0,0) \oplus (1,1)$ and thus have a trivial subrepresentation.

\begin{theorem}
    \label{thm:gravitons}
    Suppose $\Poincupd \stackrel{U}{\curvearrowright} \hi$ and let $\D$ be a dense subset of $\hi$. Let $S : SL(2,\mathbb{C}) \to GL(\mathbb{C}^n)$ be a finite-dimensional representation of $SL(2,\mathbb{C})$ such that
    \begin{equation}
        S[A] = D^{(0,0)}[A] \oplus D^{(j_1,j_2)}[A] \qquad \forall A \in SL(2,\mathbb{C})
    \end{equation}
    where $(j_1,j_2) \neq (0,0)$ is irreducible, and consider the spinor $\{\hat\psi_\alpha : \mink \to \ldh\}_{\alpha=1}^{n \in \mathbb{N}}$. If
    \begin{enumerate}
        \item $U$ is a weakly continuous unitary representation of $\Poincupd$,
        \item For all $(a,A) \in \Poincupd$, $U(a,A) \D \subset \D$ and $U(a,A)^\dagger \D \subset \D$,
        \item $\hat{\psi}_\alpha(x) \D \subset \D$ and $\hat{\psi}_\alpha(x)^\dagger \D \subset \D$ for all $x \in \mink$ and for all $\alpha \in \{1,\cdots,n\}$,
        \item For all $(a,A) \in \Poincupd, \, x \in \mink, \, \alpha \in \{1,\cdots,n\}$, \[ U(a,A) \hat{\psi}_\alpha(x) U(a,A)^\dagger = \sum_{\beta=1}^n S[A^{-1}]_{\alpha \beta} \hat{\psi}_\beta(\Lambda(A) x + a) \, , \]
        which is understood on $\D$,
        \item There exists a unique (up to constant multiples) non-zero vector $\ket{\Omega} \in \D$ such that $U(x,e) \ket{\Omega} = \ket{\Omega}$ for all $x \in \mink$.
    \end{enumerate}
    Then $\exists c \in \mathbb{C}$ such that for all $x \in \mink$,
\[
    \begin{cases}
        \hat\psi_1(x)\ket\Omega = c \ket{\Omega} \, , \quad \hat\psi_1(x)^\dagger\ket\Omega = \bar{c} \ket{\Omega} \, , \\
        \hat\psi_i(x)\ket\Omega = \hat\psi_i(x)^\dagger\ket\Omega = 0 \qquad \forall i \in \{2,\cdots,n\} \, .
    \end{cases}
\]
\sam{In particular, if $\ket{\Omega}$ is separating for the unital $*$-algebra generated by the $\hat{\psi}_\alpha$ and $\hat{\psi}_\alpha^\dagger$, then $\exists c \in \mathbb{C}$ such that for all $x \in \mink$,}
\[
    \sam{\begin{cases}
        \hat\psi_1(x)= c \mathbb{1}_\D \, , \quad \hat\psi_1(x)^\dagger = \bar{c} \mathbb{1}_\D \, , \\
        \hat\psi_i(x)= \hat\psi_i(x)^\dagger = 0 \qquad \forall i \in \{2,\cdots,n\} \, 
    \end{cases}
    }
    \]
    \sam{on $\D$.}
\end{theorem}

\begin{proof}
    See Appendix \ref{app:gravitons}. 
\end{proof}

\sam{Note once again that the covariance law assumed in Thm.~\ref{thm:chiral} does not include transformations that close only modulo gauge transformations. Thus, for electromagnetism, the clean gauge-invariant object to apply this no-go theorem on is the field strength $\hat{F}_{\mu\nu}$, i.e. a $(1,0)\oplus(0,1)$ multiplet. The vector potential $\hat{A}_\mu$ transforms as a $(1/2,1/2)$ multiplet only up to a gauge transformation. \\ 

Similarly, a symmetric rank-two field $\hat{h}_{\mu\nu}$ seen as a $(0,0)\oplus(1,1)$ falls under Thm.~\ref{thm:gravitons} only when treated as an exact Lorentz tensor multiplet. The metric perturbation in linearised gravity is gauge-dependent under linearised diffeomorphisms. For a gauge-invariant gravitational object, one should instead formulate the statement for the linearised curvature, for example the linearised Weyl multiplet $(2,0)\oplus(0,2)$, which is also covered by the argument above.} \\

\sam{Hence, exact pointwise Poincaré covariance for finite-dimensional Lorentz multiplets is highly constrained in the presence of an invariant vacuum. For multiplets with no trivial Lorentz sub-representation, Thm.~\ref{thm:chiral} shows that every component and its adjoint annihilate the invariant vacuum. For a multiplet containing one scalar component and one non-trivial irreducible component, Thm.~\ref{thm:gravitons} shows that only the scalar component can act non-trivially on the vacuum, and even then only as a constant multiple. These statements are raised at the operator level after assuming that the vacuum is separating for the relevant field algebra.} \\

Though an extension which may or may not be expected from the scalar case covered by Wizimirski's theorem, we find this interesting: it confirms the epistemic rather than ontological nature of Poincaré covariance, and underlines the difficulty of quantising Lorentzian metrics pointwise. \\

Whether this result can be further extended to more general covariance groups -- e.g. the isometry group of some other Lorentzian spacetimes of interest (though those are not always locally compact so need not admit the usual representation theory results on Hilbert spaces) -- is an interesting open question. \\ 

\sam{We see weak continuity as the main cause of the no-go theorems related to (anti)commutation relations and causality; in particular, we expect pointwise quantum fields which are not weakly continuous (e.g. relational local quantum fields found in Ref.~\cite{fedida_foundations_2025}) to avoid these obstructions. \\ 

The no-go results related to pointwise covariance are not immediately related to weak continuity however. Indeed, Thm. \ref{thm:chiral} does not require any such notions of continuity. Rather, we expect the notion of covariance to be too strong to be implemented pointwisely: after all, it is an epistemic requirement (the physics looks the same in all inertial reference frames) rather than an ontic one. Again, relational local quantum fields \cite{fedida_foundations_2025} are not, at the functional analytic level, pointwise covariant (but rather also have a shift in the state of the quantum reference frame), and thus avoid these results.}

\sam{
\section*{Declaration of generative AI and AI-assisted technologies in the manuscript preparation process}

The author declares that no generative AI or AI-assisted technologies were used in the preparation of this manuscript.
}

\section*{Acknowledgements}

The author thanks Adrian Kent for interesting discussions, Jan Głowacki for useful feedback on the first draft of this paper\sam{, and an anonymous reviewer for clarifying the proof that a finite Lorentz-invariant measure must be concentrated at the origin}. The author is funded by a studentship from the Engineering and Physical Sciences Research Council, Grant No. 2882481. The author states that there is no conflict of interest and declares that the data supporting the findings of this study are available within the paper.

\onecolumngrid
\bibliographystyle{unsrt}
\bibliography{references,references-2}

\begin{appendix}
    \section{No-go -- spinor covariance for nontrivial representations}

    \label{app:nontrivial rep}

    We here provide a proof for Theorem \ref{thm:chiral}. Let
    \begin{align}
        f_{\alpha \beta}(x) &:= \expval{\hat{\psi}_\alpha(0) \hat{\psi}_\beta(x)^\dagger}{\Omega} \\
        g_{\alpha \beta}(x) &:= \expval{\hat{\psi}_\alpha(0)^\dagger \hat{\psi}_\beta(x)}{\Omega}
    \end{align}
    for all $x \in \mink$. We have that for all $A \in SL(2,\mathbb{C})$ and all $x \in \mink$,
    \begin{equation}
        U(0,A^{-1}) \hat{\psi}_\alpha(0) U(0,A^{-1})^\dagger = \sum_{\beta=1}^n S[A]_{\alpha \beta} \hat{\psi}_\beta(0) 
    \end{equation} \\
    and
    \begin{equation}
        U(0,A^{-1})\hat{\psi}_\alpha(\Lambda(A) x)^\dagger U(0,A^{-1})^\dagger = \sum_{\beta=1}^n \overline{S[A]_{\alpha \beta}}  \hat{\psi}_\beta(x)^\dagger 
    \end{equation}
    so
    \begin{equation}
        \hat{\psi}_\alpha(\Lambda(A) x)^\dagger = \sum_{\beta=1}^n \overline{S[A]_{\alpha \beta}} U(0,A^{-1})^\dagger \hat{\psi}_\beta(x)^\dagger U(0,A^{-1}) \, .
    \end{equation}
    Thus,
    \begin{align}
        f_{\alpha \beta}(\Lambda(A) x) &= \expval{\hat{\psi}_\alpha(0)\hat{\psi}_\beta(\Lambda(A) x)^\dagger}{\Omega} \\
        &= \sum_{\delta=1}^n \overline{S[A]_{\beta \delta}} \expval{U(0,A^{-1})^\dagger  U(0,A^{-1})\hat{\psi}_\alpha(0) U(0,A^{-1})^\dagger \hat{\psi}_\delta(x)^\dagger U(0,A^{-1})}{\Omega} \\
        &= \sum_{\gamma,\delta = 1}^n S[A]_{\alpha \gamma} \overline{S[A]_{\beta \delta}} \expval{\hat{\psi}_\gamma(0) \hat{\psi}_\delta(x)^\dagger}{\Omega} = \sum_{\gamma,\delta = 1}^n S[A]_{\alpha \gamma} \overline{S[A]_{\beta \delta}} f_{\gamma \delta}(x) \, .
        \end{align}
    Hence, writing $F(x) = [f_{\alpha \beta}(x)]$, we have that for all $x \in \mink$ and $A \in SL(2,\mathbb{C})$,
    \begin{equation}
        F(\Lambda(A) x) = S[A] F(x) S[A]^\dagger
    \end{equation}
    so, in particular,
    \begin{equation}
        \label{eqn:F eqn}
        F(0) = S[A] F(0) S[A]^\dagger \qquad \forall A \in SL(2,\mathbb{C}) \, .
    \end{equation}
    Likewise, writing $G(x) = [g_{\alpha \beta}(x)]$, we have that
    \begin{equation}
        \label{eqn:G eqn}
        G(0) = \overline{S[A]} G(0) S[A]^T \qquad \forall A \in SL(2,\mathbb{C}) \, .
    \end{equation}
    Note that both $F(0)$ and $G(0)$ are positive semi-definite matrices. Indeed, take any $\kappa = (\kappa_1,\cdots,\kappa_n) \in \mathbb{C}^n$ and let
    \begin{equation}
        X := \sum_{\beta =1}^n \kappa_\beta \hat{\psi}_\beta(0) 
    \end{equation}
    Then
    \begin{equation}
        \kappa^\dagger G(0) \kappa = \sum_{\alpha, \beta = 1}^n \overline{\kappa_\alpha} g_{\alpha \beta}(0) \kappa_\beta =\expval{\left(\sum_{\alpha=1}^n \overline{\kappa_\alpha} \hat{\psi}_\alpha(0)^\dagger\right)\left(\sum_{\beta=1}^n \kappa_\beta \hat{\psi}_\beta(0)\right)}{\Omega} = \expval{X^\dagger X}{\Omega} = \norm{X \ket{\Omega}}^2 \geq 0
    \end{equation}
    so $G(0) \geq 0$, and likewise $F(0) \geq 0$. Let us consider constraints on $F(0)$ -- the case of $G(0)$ follows analogously. Multiplying eqn. \eqref{eqn:F eqn} on the right by $(S^\dagger)^{-1}$ yields
    \begin{equation}
        S[A] F(0) = F(0) (S[A]^\dagger)^{-1} \qquad \forall A \in SL(2,\mathbb{C}) \, .
    \end{equation}
    Suppose for now that $S$ is irreducible. For any $z \in \text{ran } F(0)$ there exists $y \in \mathbb{C}^n$ such that $z = F(0) y$ and so
    \begin{equation}
        S[A] z = S[A](F(0) y) = (S[A] F(0)) y = F(0) ((S[A]^\dagger)^{-1} y) \in \text{ran } F(0)
    \end{equation}
    for all $A \in SL(2,\mathbb{C})$, so $S[A](\text{ran } F(0)) \subseteq \text{ran } F(0)$ for all $A \in SL(2,\mathbb{C})$ so $\text{ran } F(0)$ is $S$-invariant. Since $S$ is irreducible, the only invariant subspaces are $\{0\}$ and $\mathbb{C}^n$. If $\text{ran } F(0) = \{0\}$ then $F(0) = 0$ and we are done. Assume for contradiction that $\text{ran } F(0) = \mathbb{C}^n$. Then $F(0)$ is surjective (and finite dimensional) so by the rank-nullity theorem, it is bijective. A bijective positive semi-definite matrix is necessarily positive definite (if $v \neq 0$ and $v^\dagger F(0) v = 0$ then $\norm{F(0)^{1/2} v}^2 = 0 \Rightarrow F(0)^{1/2} v = 0 \Rightarrow v \in \ker F(0) = \{0\}$ which is a contradiction). Hence, $F(0) > 0$ and the unique positive square root $F(0)^{1/2}$ exists and is invertible. Define 
    \begin{equation}
        U : A \in SL(2,\mathbb{C}) \mapsto F(0)^{-1/2} S[A] F(0)^{1/2} \in GL(\mathbb{C}^n) \, .
    \end{equation}
    Then
    \begin{multline}
        U[A] U[A]^\dagger = F(0)^{-1/2} S[A] F(0)^{1/2} F(0)^{1/2} S[A]^\dagger F(0)^{-1/2} = F(0)^{-1/2} S[A] F(0) S[A]^\dagger F(0)^{-1/2}\\  = F(0)^{-1/2} F(0) F(0)^{-1/2} = \mathbb{1}_{n \times n}
    \end{multline}
    so each $U[A]$ is unitary. Moreover, for all $A,B \in SL(2,\mathbb{C})$,
    \begin{equation}
        U[AB] = F(0)^{-1/2} S[AB] F(0)^{1/2}=(F(0)^{-1/2}S[A] F(0)^{1/2})(F(0)^{-1/2}S[B] F(0)^{1/2}) = U[A] U[B]
    \end{equation}
    \sam{Thus $U$ is a finite-dimensional unitary representation of $SL(2,\mathbb C)$. Since $SL(2,\mathbb C)$ has no nontrivial finite-dimensional unitary representations, $U[A]=\mathbb{1}_{n\times n}$ for all $A$. Hence $S[A]=\mathbb{1}_{n\times n}$ for all $A$, contradicting the assumption that $S$ is a nontrivial irreducible representation. Therefore $\text{ran }F(0) \neq \mathbb{C}^n$, and irreducibility thus forces $F(0) = 0$.} Likewise, eqn. \eqref{eqn:G eqn} yields $G(0) = 0$ following the same argument. Therefore $F(0)=G(0)=0$. From $G(0)=0$ we obtain, for each $\alpha$,
\begin{equation}
    \|\hat\psi_\alpha(0)\ket\Omega\|^2=\langle\Omega|\hat\psi_\alpha(0)^\dagger\hat\psi_\alpha(0)|\Omega\rangle=g_{\alpha\alpha}(0)=0 \, ,
\end{equation}
thus $\hat\psi_\alpha(0)\ket\Omega=0$. Similarly $F(0)=0$ yields $\hat\psi_\alpha(0)^\dagger\ket\Omega=0$. Finally, for every $x\in\mink$,
\begin{equation}
    \hat\psi_\alpha(x)\ket\Omega
=U(x,e)\,\hat\psi_\alpha(0)\,U(x,e)^\dagger\ket\Omega
 =U(x,e)\,\hat\psi_\alpha(0)\ket\Omega
=0 \, ,
\end{equation}
and likewise from $F(0)=0$ we get $\hat\psi_\alpha(x)^\dagger\ket\Omega=0$ so the proof is shown for nontrivial irreducible representations of $SL(2,\mathbb{C})$. More generally, we decompose the finite-dimensional representation $S$ into irreducibles as
\begin{equation}
    S[A]=\bigoplus_{k=1}^r  S^{(k)}[A]\, ,\quad S^{(k)}[A]\in GL(\mathbb{C}^{d_k}) \, ,\quad \sum_{k=1}^r d_k = n \, .
\end{equation}
Choose a basis of $\mathbb{C}^n$ so that the indices corresponding to each sub-representation $S^{(k)}$ form a contiguous block, i.e.
\begin{equation}
    \mathbb{C}^n \cong \bigoplus_{k=1}^r  \mathbb{C}^{d_k} \, ,
\end{equation}
and the representation matrices $S[A]$ are block-diagonal with $r$ blocks $S^{(k)}[A]$ on the diagonal.
For fixed $k$, consider the sub-multiplet of $\{\hat{\psi}_\alpha\}$ whose indices lie in the block corresponding to $S^{(k)}$. These components inherit a Poincaré covariance law under $U$ with the finite-dimensional representation $S^{(k)}$, but restricted to the relevant indices. Since each $S^{(k)}$ is irreducible we have from the above argument that for each $\alpha$ in that block and every $x \in \mink$,
\begin{equation}
    \hat\psi_\alpha(x)\ket\Omega = \hat\psi_\alpha(x)^\dagger\ket\Omega = 0 \, .
\end{equation}
Repeating this for all $k \in \{1,\cdots,r\}$ covers all the components of the original multiplet and thus completes the \sam{first part of the proof. The last part of the proof follows by definition of a separating vector: $\hat{\psi}_\alpha(x) \ket{\Omega} = \hat{\psi}_\alpha(x)^\dagger \ket{\Omega} = 0 \Rightarrow \hat{\psi}_\alpha(x) = \hat{\psi}_\alpha(x)^\dagger = 0$.}

\section{No-go -- spinor covariance for gravitons}

\label{app:gravitons}

We start with a lemma that will prove to be useful for the following.

\sam{\begin{lemma}
\label{lem:invariant measure}
Let $\mu$ be a finite positive Borel measure on $\mathbb R^4$
which is invariant under the proper orthochronous Lorentz group.
Then $\mu$ is concentrated at the origin.
\end{lemma}

\begin{proof}
For $i\in\{1,2,3\}$ and $k\in\mathbb Z$, let
\begin{equation}
A_{i,k}^{\pm} := \left\{p\in\mathbb R^4:2^k\leq |p^0\pm p^i|<2^{k+1}\right\}.
\end{equation}
These sets are Borel, since they are inverse images of the Borel intervals $[2^k,2^{k+1})$ under the continuous maps $p\mapsto |p^0\pm p^i|$. Let $B_i(t)$ be the boost of rapidity $t$ in the $i$-th spatial direction, with
\begin{equation}
(B_i(t)p)^0=\cosh(t)p^0+\sinh(t)p^i,\qquad
(B_i(t)p)^i=\sinh(t)p^0+\cosh(t)p^i ,
\end{equation}
and with the remaining components unchanged. Then
\begin{equation}
(B_i(t)p)^0\pm(B_i(t)p)^i
=
e^{\pm t}(p^0\pm p^i).
\end{equation}
Therefore, for $t=n\log2$,
\begin{equation}
    \abs{        (B_i(n\log2)p)^0\pm(B_i(n\log2)p)^i} = 2^{\pm n}\abs{p^0\pm p^i}
\end{equation}
and likewise
\begin{equation}
    2^{\pm n}\abs{(B_i(-n\log2)p)^0\pm (B_i(-n\log2)p)^i} = \abs{p^0\pm p^i}.
\end{equation}
Hence, for every $n\in\mathbb Z$,
\begin{equation}
B_i(n\log2)A_{i,k}^{\pm}
=
A_{i,k\pm n}^{\pm}.
\end{equation}
For fixed $i$, $k$ and sign, these sets are pairwise disjoint as
$n$ varies, since the half-open intervals
$[2^j,2^{j+1})$ are pairwise disjoint. Lorentz invariance
therefore gives $\mu(E) = \mu(\Lambda \cdot E)$ for all $E \in \Bor(\mathbb{R}^4)$ and $\Lambda \in SO_+^\uparrow(1,3)$, so for every $N\in\mathbb N$,
\begin{equation}
(2N+1)\mu(A_{i,k}^{\pm})= \sum_{n=-N}^{N} \mu\left(B_i(n\log2)A_{i,k}^{\pm}\right)= \mu\left(\bigsqcup_{n=-N}^{N} B_i(n\log2)A_{i,k}^{\pm} \right)\leq\mu(\mathbb R^4).
\end{equation}
Since $\mu(\mathbb R^4)<\infty$ and this holds for every $N$, we obtain $\mu(A_{i,k}^{\pm})=0$.

Finally, if $p^0\pm p^i=0$ for every $i$ and both signs, then adding and subtracting the two equations $p^0+p^i=0$ and $p^0-p^i=0$ shows that $p^0=p^i=0$ for every $i$. Hence every nonzero $p$ has $|p^0\pm p^i|>0$ for some $i$ and sign, and this positive number belongs to a unique half-open interval $[2^j,2^{j+1})$. Thus,
\begin{equation}
\mathbb R^4\smallsetminus\{0\} = \bigcup_{i=1}^{3} \bigcup_{\sigma=\pm1} \bigcup_{k\in\mathbb Z} A_{i,k}^{\sigma}.
\end{equation}
Countable subadditivity thus gives $\mu(\mathbb R^4\smallsetminus\{0\})=0$, i.e. $\mu=\mu(\{0\})\delta_0$.
\end{proof}}

We now provide a proof for Theorem \ref{thm:gravitons}. Let
    \begin{align}
        f_{\alpha \beta}(x) &:= \expval{\hat{\psi}_\alpha(0) \hat{\psi}_\beta(x)^\dagger}{\Omega} \\
        g_{\alpha \beta}(x) &:= \expval{\hat{\psi}_\alpha(0)^\dagger \hat{\psi}_\beta(x)}{\Omega}
    \end{align}
    for all $x \in \mink$. Note now that we have a unique non-zero translation invariant vacuum. But for all $A \in SL(2,\mathbb{C})$,
    \begin{equation}
        U(0,A) \ket{\Omega} = U(0,A) U(x,e) \ket{\Omega} = U(\Lambda(A)x,e) U(0,A) \ket{\Omega}
    \end{equation}
    so by uniqueness (and unitarity) $\exists \alpha(A) \in [0,2\pi)$ such that $U(0,A) \ket{\Omega} = e^{i \alpha(A)} \ket{\Omega}$. Thus, as in Theorem \ref{thm:chiral}, we have
    \begin{align}
        f_{\alpha \beta}(\Lambda(A) x) &= \expval{\hat{\psi}_\alpha(0)\hat{\psi}_\beta(\Lambda(A) x)^\dagger}{\Omega} \\
        &= \sum_{\delta=1}^n \overline{S[A]_{\beta \delta}} \expval{U(0,A^{-1})^\dagger  U(0,A^{-1})\hat{\psi}_\alpha(0) U(0,A^{-1})^\dagger \hat{\psi}_\delta(x)^\dagger U(0,A^{-1})}{\Omega} \\
        &= e^{-i \alpha(A^{-1})} e^{i \alpha(A^{-1})} \sum_{\gamma,\delta = 1}^n S[A]_{\alpha \gamma} \overline{S[A]_{\beta \delta}} \expval{\hat{\psi}_\gamma(0) \hat{\psi}_\delta(x)^\dagger}{\Omega} = \sum_{\gamma,\delta = 1}^n S[A]_{\alpha \gamma} \overline{S[A]_{\beta \delta}} f_{\gamma \delta}(x) \, ,
        \end{align}
    and likewise for $g_{\alpha \beta}$. Thus, we once again have
    \begin{align}
        \label{eqns: F and G F}
        F(\Lambda(A) x) &= S[A] F(x) S[A]^\dagger \\
        \label{eqns: F and G G}
        G(\Lambda(A)x) &= \overline{S[A]} G(x) S[A]^T
    \end{align}
    for all $x \in \mink$ and all $A \in SL(2,\mathbb{C})$, where $F(x) = [f_{\alpha \beta}(x)]$ and $G(x) = [g_{\alpha \beta}(x)]$.
    Defining $\tilde{F}(x) = [f_{ij}(x)]$ and $\tilde{G}(x) = [g_{ij}(x)]$ for $i,j \in \{2,\cdots,n\}$, we find
    \begin{align}
        \tilde{F}(\Lambda(A)x) &= D^{(j_1,j_2)}[A] \tilde{F}(x) D^{(j_1,j_2)}[A]^\dagger \\
        \tilde{G}(\Lambda(A)x) &= \overline{D^{(j_1,j_2)}[A]} \tilde{G}(x) D^{(j_1,j_2)}[A]^T
    \end{align}
    Following the same line of argument as in Theorem \ref{thm:chiral}, we have $\tilde{G}(0) = \tilde{F}(0) = 0$ so
    \begin{equation}
        \hat\psi_i(x)\ket\Omega = \hat\psi_i(x)^\dagger\ket\Omega = 0 
    \end{equation}
    for all $x \in \mink$ and all $i \in \{2,\cdots,n\}$. From eqns. \eqref{eqns: F and G F} and \eqref{eqns: F and G G}, we then have
    \begin{align}
        f_{11}(\Lambda(A)x) &= f_{11}(x) \\
        g_{11}(\Lambda(A)x) &= g_{11}(x)
    \end{align}
    for all $x \in \mink$ and all $A \in SL(2,\mathbb{C})$. But 
    \begin{align}
        f_{1 1}(x) &= \expval{\hat{\psi}_1(0) U(x,e) \hat{\psi}_1(0)^\dagger}{\Omega} \\
        g_{1 1}(x) &= \expval{\hat{\psi}_1(0)^\dagger U(x,e) \hat{\psi}_1(0)}{\Omega} \, .
    \end{align}
    \sam{Let $\ket{\zeta_f}:=\hat\psi_1(0)^\dagger\ket{\Omega}$ and $\ket{\zeta_g}:=\hat\psi_1(0)\ket{\Omega}$. Then $f_{11}(x) = \expval{U(x,e)}{\zeta_f}$ and $g_{11}(x) = \expval{U(x,e)}{\zeta_g}$. Thus \(f_{11}\) and \(g_{11}\) are continuous functions of positive type, bounded respectively by $\norm{\zeta_f}^2$ and $\norm{\zeta_g}^2$.} By Bochner's theorem \cite{rudin_chapter_1996}, there exists a finite positive Borel measure $\mu_{11} : \Bor(\mathbb{R}^4) \to \mathbb{R^+}$ such that
        \begin{equation}
            f_{11}(x) = \int_{\mathbb{R}^4} e^{i p \cdot x} d\mu_{11}(p) \, .
        \end{equation}
        The covariance property of $f_{11}(x)$ in turn propagates to $\mu_{11}$, since
        \begin{multline}
        f_{11}(\Lambda(A) x) = \int_{\mathbb{R}^4} e^{i p \cdot \Lambda(A) x} d\mu_{11}(p) 
        = \int_{\mathbb{R}^4} e^{i \Lambda(A)^{-1} p \cdot x} d\mu_{11}(p) \stackrel{q := \Lambda(A)^{-1} p}{=}\int_{\mathbb{R}^4} e^{i q\cdot x} d[(\Lambda(A)^{-1})_* \mu_{11}](q) \\\stackrel{\mu^\Lambda_{11}:= (\Lambda(A)^{-1})_* \mu_{11}}{=} \int_{\mathbb{R}^4} e^{i p \cdot x} d\mu^\Lambda_{11}(p) \stackrel{!}{=} \int_{\mathbb{R}^4} e^{i p \cdot x} d\mu_{11}(p)
    \end{multline}
    which by the uniqueness of Fourier transforms implies that $\mu_{11}(\Lambda(A) \cdot X) = \mu_{11}(X)$ for all $X \in \Bor(\mathbb{R}^4)$ and all $A \in SL(2,\mathbb{C})$. \sam{By Lemma \ref{lem:invariant measure},} this implies that
\begin{equation}
    \mu_{11}
    =
    \mu_{11}(\{0\})\,\delta_0,
\end{equation}
and thus $f_{11}(x)
    =
    \mu_{11}(\{0\})$ is constant. Applying the same argument to the finite positive measure representing $g_{11}$ shows that $g_{11}$ is also constant.
    
    Let $\ket{\eta}:= \hat{\psi}_1(0) \ket{\Omega}$. Then $g(x) = g(0)$ is expressed as
    \begin{equation}
        \braket{\eta}{U(x,e) \eta} = \braket{\eta} = \norm{\eta}^2
    \end{equation}
    But we also have $\norm{\eta} = \norm{U(x,e) \eta}$ since $U(x,e)$ is unitary, so
    \begin{equation}
        \braket{\eta}{U(x,e) \eta} = \norm{\eta} \cdot \norm{U(x,e) \eta}
    \end{equation}
    which by the Cauchy-Schwarz inequality implies that $U(x,e) \ket{\eta} = \ket{\eta}$ for all $x \in \mink$. Hence, $U(x,e) \hat{\psi}_1(0) \ket{\Omega} = \hat{\psi}_1(0) \ket{\Omega}$ and $\hat{\psi}_1(x) \ket{\Omega} = U(x,e) \hat{\psi}_1(0) \ket{\Omega}$ so 
    \begin{equation}
        U(y,e) \hat{\psi}_1(x) \ket{\Omega} =  U(y,e) U(x,e)\hat{\psi}_1(0) \ket{\Omega} = U(x,e) U(y,e) \hat{\psi}_1(0) \ket{\Omega} = U(x,e) \hat{\psi}_1(0) \ket{\Omega} = \hat{\psi}_1(x) \ket{\Omega} \, .
    \end{equation}
    Thus, all vectors $\hat{\psi}_1(x) \ket{\Omega}$ are invariant under the translation group for all $x \in \mink$, which by uniqueness implies that $\exists c \in \mathbb{C}$ such that 
    $\forall x \in \mink$, $\hat{\psi}_1(x) \ket{\Omega} = c \ket{\Omega}$. This $c$ is position independent since, for any $x,a \in \mink$, $\hat{\psi}_1(x+a) \ket{\Omega} = U(a,e) \hat{\psi}_1(x) U(a,e)^\dagger \ket{\Omega} = U(a,e) c(x) \ket{\Omega} = c(x) \ket{\Omega}$ so $c(x+a) = c(x) \equiv c$. The same reasoning holds for $\hat{\psi}_1(x)^\dagger$, so that $\hat{\psi}_1(x)^\dagger \ket{\Omega} = d \ket{\Omega}$ for some $d \in \mathbb{C}$. But for normalised $\ket{\Omega}$,
    \begin{equation}
        d = \expval{\hat{\psi}_1(x)^\dagger}{\Omega} = \braket{\hat{\psi}_1(x) \Omega}{\Omega} = \braket{c \Omega}{\Omega} = \bar{c} \, ,
    \end{equation}
    which concludes the \sam{first part of the proof. The last part of the proof follows by definition of a separating vector: in particular, it is separating for the \emph{unital} $*$-algebra generated by the $\hat{\psi}_\alpha$, i.e. $(\hat{\psi}_1(x) - c \mathbb{1}_\D) \ket{\Omega} = 0 \Rightarrow \hat{\psi}_1(x) = c\mathbb{1}_\D$ and likewise for $\hat{\psi}_1(x)^\dagger = \bar{c} \mathbb{1}_\D$, and $\hat{\psi}_i(x) \ket{\Omega} = \hat{\psi}_i(x)^\dagger \ket{\Omega} = 0 \Rightarrow \hat{\psi}_i(x) = \hat{\psi}_i(x)^\dagger = 0 $.}
\end{appendix}

\end{document}